\documentclass[a4paper,USenglish,cleveref]{lipics-v2021}
\nolinenumbers
 \category{} 

\usepackage[export]{adjustbox}
\usepackage{booktabs}
\usepackage[group-separator={,}]{siunitx}
\usepackage{multicol}
\usepackage{amsthm}
\usepackage{amsmath}
\usepackage{graphicx} \usepackage{microtype}
\usepackage{xspace}
\usepackage{tikz}

\usetikzlibrary{graphs}

\title{Hardness Results on Characteristics for Elastic-Degenerated Strings}
\date{\today}
\author{Dominik K\"{o}ppl}{University of Yamanashi, Japan \and \url{https://dkppl.de/}}{dkppl@yamanashi.ac.jp}{https://orcid.org/0000-0002-8721-4444}{JSPS KAKENHI Grant Number 23H04378 and Yamanashi Wakate Grant Number 2291}
\author{Jannik Olbrich}{Ulm University, Ulm, Germany}{jannik.olbrich@uni-ulm.de}{https://orcid.org/0000-0003-3291-7342}{Supported by the Deutsche Forschungsgesellschaft, Grant No.\ OH 53/7-1}{}
\authorrunning{D.\ Köppl and J.\ Olbrich}
\Copyright{Dominik Köppl and Jannik Olbrich}
\hideLIPIcs

\newtheorem{problem}{Problem}

\newcommand{\antipower}{{\mdseries\scshape Anti-Power}\xspace}

\newcommand{\PbAbsentWord}{{\mdseries\scshape Absent Word}\xspace}
\newcommand{\PbUniqueSubstr}{{\mdseries\scshape Unique Substring}\xspace}
\newcommand{\PbLongestRepeating}{{\mdseries\scshape lrf}\xspace}
\newcommand{\PbStringIntersection}{{\mdseries\scshape edsi}\xspace}
\newcommand{\PbHamiltonianPath}{{\mdseries\scshape Hamiltonian Path}\xspace}
\newcommand{\PbCommonSubsequence}{{\mdseries\scshape Common Subsequence}\xspace}
\newcommand{\PbLongestPrevious}{{\mdseries\scshape lpf}\xspace}

\newcommand*{\edc}[1]{\left\{\arraycolsep=1.4pt\begin{array}{c}#1\end{array}\right\}}
\newcommand*{\eds}[1]{\ensuremath{\widetilde{#1}}}
\newcommand*{\Language}{\ensuremath{\mathcal{L}}}

\newcommand*{\adj}{\mathrm{adj}}

\definecolor{solarizedYellow}{HTML}{B58900}
\definecolor{solarizedOrange}{HTML}{CB4B16}
\definecolor{solarizedRed}{HTML}{DC322F}
\definecolor{solarizedMagenta}{HTML}{D33682}
\definecolor{solarizedViolet}{HTML}{6C71C4}
\definecolor{solarizedBlue}{HTML}{268BD2}
\definecolor{solarizedCyan}{HTML}{2AA198}
\definecolor{solarizedGreen}{HTML}{859900}

\usepackage{listings}

\newcommand*{\block}[1]{\subparagraph{#1}}

\newcommand*{\StateComplexity}[1]{\begin{flushright}
\vspace{-0.5em}
\color{gray}
\{#1\}
\end{flushright}
}

\begin{document}
\maketitle
\begin{abstract}
    Generalizations of plain strings have been proposed as a compact way to represent a collection of nearly identical sequences
    or to express uncertainty at specific text positions by enumerating all possibilities.
    While a plain string stores a character at each of its positions, generalizations consider a set of characters (indeterminate strings),
    a set of strings of equal length (\emph{generalized degenerate strings}, or shortly \emph{GD strings}), 
    or a set of strings of arbitrary lengths (\emph{elastic-degenerate strings}, or shortly \emph{ED strings}).
    These generalizations are of importance to compactly represent such type of data, and find applications in bioinformatics for representing and maintaining a set of genetic sequences of the same taxonomy or a multiple sequence alignment.
    To be of use, attention has been drawn to answering various query types such as pattern matching or measuring similarity of ED strings by generalizing techniques known to plain strings.
    However, for some types of queries, it has been shown that a generalization of a polynomial-time solvable query on classic strings becomes NP-hard on ED strings, e.g.~[Russo et al.,2022].
    In that light, we wonder about other types of queries, which are of particular interest to bioinformatics: the search for the longest repeating factor, unique substrings, 
    absent words, anti-powers, and longest previous factors. 
    While we obtain a polynomial time algorithm for the first problem on ED strings, we show that all others are NP-hard to compute, some of them even under the restriction that the input can be modelled as an indeterminate or GD string.
\end{abstract}

\section{Introduction}

In many applications, uncertainty is a known problem that hinders the adaptation of classic algorithms requiring full information.
We here model uncertainty by explicitly stating all possibilities.
Such a model found popularity when working with biological data such as gene sequences, for which indeterminate strings~\cite{smyth09adaptive} have been proposed.
An indeterminate string models a string that can have multiple alternative characters at any of its positions; 
each element of an indeterminate string becomes a subset of the input alphabet, which we call a \emph{symbol}.
Symbols generalize the IUPAC notation~\cite{iupac71abbreviations} to represent the ambiguous specification of nucleic acids in DNA or RNA sequences.
A generalization of indeterminate strings are generalized degenerate (GD) strings that store strings of equal length instead of characters.
Finally, \cite{iliopoulos21efficient} removed the restriction on equal length and called the modelled string an elastic-degenerate (ED) string,
which was proposed as a model of a set of similar sequences.
From a language theoretic point of view, an ED-string is a regular expression without nested parentheses or Kleene star.
\Cref{figExampleEDString} gives an example of the three models.

\begin{figure}[h]
    \begin{minipage}{0.25\linewidth}
\[
    \eds{T_1} = 
    {\texttt{A}}
\edc{\texttt{A} \\ \texttt{C}}
{\texttt{C}}
\edc{\texttt{G} \\ \texttt{T}}
\]
$
\Language(\eds{T_1}) = \{$
\texttt{AACG},
\texttt{AACT},
\texttt{ACCG},
\texttt{ACCT}
$\}$
    \end{minipage}
    \hfill
    \begin{minipage}{0.25\linewidth}
\[
    \eds{T_2} = 
    {\texttt{A}}
\edc{\texttt{A}\texttt{G} \\ \texttt{C}\texttt{A}}
    {\texttt{C}}
\edc{\texttt{G} \\ \texttt{T}}
\]
$\Language(\eds{T_2}) = \{$
\texttt{AAGCG},
\texttt{AAGCT},
\texttt{ACACG},
\texttt{ACACT}
$\}$
    \end{minipage}
    \hfill
    \begin{minipage}{0.35\linewidth}
\[
    \eds{T_3} = 
    {\texttt{A}}
\edc{\texttt{A}\texttt{G} \\ \epsilon}
    {\texttt{C}}
\edc{\texttt{G} \\ \texttt{T} \\ \texttt{C}\texttt{C}\texttt{T}}
\]
$\Language(\eds{T_3}) = \{$
\texttt{AAGCG},
\texttt{AAGCT},
\texttt{AAGCCCT},
\texttt{ACG},
\texttt{ACT},
\texttt{ACCCT}
$\}$
    \end{minipage}
    \caption{Example for an indeterminate string~$\eds{T_1}$ (left), 
    a GD string~$\eds{T_2}$ (middle), and 
an ED string~$\eds{T_3}$ (right). The language $\Language(\eds{T_i})$ of each string $\eds{T_i}$ is depicted below; the language is the set of strings represented by the indeterminate/GD/ED string. A symbol storing only one character $c$ is written as $\edc{c}$ or just $c$.}
    \label{figExampleEDString}
\end{figure}

ED strings found application to encode the consensus of a population of sequences~\cite{liao23draft,alzamel18degenerate,buchler23efficient} in a multiple sequence alignment.
While advances in pattern matching on ED strings spawned recent research interest, not much is known about the computational complexity of 
typical string characteristics used, among others, in bioinformatics on classic strings.
We here study the adaptation of five string problems to ED strings.
For four among them, we however proof NP-hardness, which makes it unlikely to come up with efficient solutions for these problems on large ED strings.
Our results are as follows.

\begin{enumerate}
    \item We can compute the longest repeating factor in quadratic time on an ED string (Thm.~\ref{thm_longest_repeating_substring}).
    \item Finding a unique substring of length $k$ in an indeterminate string is NP-hard (Thm.~\ref{thm_mus_np_hard}).
    \item Finding an absent word of length $k$ in an indeterminate string is NP-hard (Thm.~\ref{thm_maw_np_hard}).
    \item Finding a $k$-anti-power in a GD string is NP-hard (Thm.~\ref{thm_antipower_np_hard}).
    \item Computing the longest previous factor length in an ED string is NP-hard (Thm.~\ref{thm_lpf_np_hard}).
\end{enumerate}
While we can solve the first problem in quadratic time, 
we give hardness results for the latter four, where the last problem can be reduced to the first problem for a restricted class of ED strings,
including GD and indeterminate strings.
Our reductions for NP-hardness are based on 3-SAT, \PbHamiltonianPath{}, and \PbCommonSubsequence{}.
For finding a minimal unique substring or a minimal absent word, we propose MAX-SAT formulations and an implementation that can solve both problems in practice on inputs of moderate size.
Finally, we conclude with an open problem on the inequality testing of two ED strings.

\section{Related Work}

We structure the related work in work dedicated to ED strings and work tackling one of the aforementioned problems, however on different types of input (mostly classic strings).

\subsection{ED Strings}
A big part of research on ED strings has been devoted to pattern matching, 
structural properties or regularities, the reconstruction of ED strings from indexes that are not necessarily self-indexes, and the comparison of two ED strings.

\block{Pattern Matching}
For indeterminate strings, we point out a Boyer--Moore adaptation~\cite{holub08fast},
a combination~\cite{smyth09adaptive} of ShiftAnd and Boyer--Moore--Sunday,
and an KMP-based approach~\cite{neerja20simple}.
For ED strings, there is a rather long line on improvements for exact pattern matching~\cite{aoyama18faster,cislak18sopang,bernardini19elastic,iliopoulos21efficient,prochazka23backward,ascone24unifying},
or with one error~\cite{bernardini22elastic}.
There are indices for pattern matching on ED texts, based on the suffix tree~\cite{gibney20efficient}, on the Burrows--Wheeler transform (BWT)~\cite{maciuca16natural} or 
on a graph for read mapping~\cite{buchler23efficient}.
The first study~\cite{gibney20efficient} also gives lower bounds on the time complexity for indexed pattern matching queries.
The problem to align a pattern to a GD string~\cite{mwaniki23fast} or ED string~\cite{mwaniki22optimal} has also recently been studied.

Given that $r$ is the maximum size of the set a symbol in the ED string represents,
NP-hardness for order-preserving pattern matching has been shown for $r=3$~\cite{russo22orderpreserving} and subsequently even for $r=2$~\cite{gawrychowski20indeterminate},
who also give NP-hardness for parameterized pattern matching for $r=2$. 

\block{Structural Properties}
For structural properties on indeterminate strings, we are aware of 
the computation of covers and/or seeds~\cite{antoniou08conservative,bari09finding}, 
maximal gapped palindromes~\cite{alzamel23maximal},
a study on periodicity~\cite{gabory23periodicity}, 
an extension of the Lyndon factorization~\cite{daykin17indeterminate},
and an augmentation with rank/select data structures~\cite{bille24rank}.
\cite{crochemore17covering} showed that computing the shortest cover is NP-complete, but fixed-parameter tractable (FPT) in the number of symbols with sizes larger than one.
Finally, \cite{louza21new} gave a new representation model for indeterminate strings.

\block{Reconstruction}
Reconstruction of indeterminate strings from data structures is another active line of research.
The reconstruction has been considered from border array, or suffix array and the LCP array~\cite{nazeen12indeterminate},
from a graph based on the prefix array~\cite{alatabbi15inferring},
or from a graph whose vertices are text positions and edges are text positions having matching characters~\cite{helling18constructing}.
\cite{christodoulakis15indeterminate} gave a characterization of whether an array can be the prefix array of a classic string or an indeterminate string.
Finally, \cite{blanchetsadri17new} studied the relation of prefix arrays of indeterminate strings to undirected graphs and border arrays.

\block{Comparison}
\cite{alzamel18degenerate} studied how to compare GD strings.
An extension to ED strings is due to~\cite{gabory23comparing}, who focused on the intersection of the languages of two ED strings, 
a problem they coined as the \emph{ED String Intersection Problem} (\PbStringIntersection{}).
They showed that \PbStringIntersection{} is NP-complete if the alphabet is unary while the letters are stored run-length compressed~\cite[Thm.~9]{gabory23comparing}.
Otherwise, \PbStringIntersection{} is polynomial~\cite[Thm~24]{gabory23comparing}.
Additionally, they gave conditional lower bounds for computing the intersection.

\subsection{Addressed Problems on Classic Strings}

We here give motivation for the problems we address in this article by highlighting, for each problem, its importance and existing solutions, which however do not cover ED strings.

\block{Longest Repeating Substring}
The search for longest repeating substring is important, among others, for detecting similarities in biological sequences~\cite{becher09efficient}.
While textbook solutions for this problem exist~\cite[APL4]{gusfield97algorithms},
the problem has been extended such that the reported substring has to cover specific string positions~\cite{tian15longest,xu16stabbing}
or needs to be common among several input texts~\cite{ohlebusch08space}.
While the textbook solution uses suffix trees, ideas have been adopted to BWT-based indexes~\cite{beller12spaceefficient,kulekci12efficient}.

\block{Minimal Unique Substring}
The problem of computing minimal unique substrings (MUSs) was introduced by Pei et al.~\cite{pei13sus}, 
and has found application in sequence alignment~\cite{adas15nucleotide}, genome comparison~\cite{haubold05genome}, and phylogenetic tree construction~\cite{chairungsee14new}.
A weaker version, the shortest unique substring (SUS), asks for only local minimality, which can be expressed by a point or region to cover.
The interest in computing SUSs resulted in a line of research improvements on the time and space complexities~\cite{hu14sus,tsuruta14sus,ileri15simple,mieno16sus,mieno20space}.
There are solutions offering a trade-off in both complexities~\cite{ganguly17sus,bannai20more}.
Also, different settings like the computation within a sliding window~\cite{mieno20minimal} or on run-length encoded strings~\cite{mieno16sus} have been proposed.
For theoretical guarantees, \cite{mieno17tight} studied the maximum number of SUSs covering a given query position.
Recently, a survey article~\cite{abedin20survey} has been dedicated to SUS computation.
Finally, there are variations for computing SUSs within a given range~\cite{abedin19rangesus,abedin20efficient} or computing 
SUSs with $k$ mismatches~\cite{hon17inplace,schultz21parallel,allen21ultafast}.

\block{Minimal Absent Word}
Minimal absent words (MAWs) have been introduced by \cite{pinho09finding} as biomarkers for potential preventive and curative medical applications.
MAWs have been found valuable for phylogeny~\cite{chairungsee12minimal}, sequence comparison~\cite{crochemore16lineartime},
information retrieval of musical content~\cite{crawford18searching},
and the reconstruction of circular binary strings~\cite{ota23reconstruction}.
For computing MAWs, algorithms have been proposed that use the suffix array~\cite{barton14lineartime}, the directed acyclic word graph (DAWG)~\cite{fujishige16computing,fujishige18truncated},
or maximal repeats~\cite{azmi16identifying}.
Algorithms working in parallel~\cite{barton15parallelising} or in external memory~\cite{heliou17emmaw} have also been proposed.
Extensions are the computation of MAWs of run-length compressed strings~\cite{akagi22minimal}, MAWs common to multiple strings~\cite{okabe23lineartime}, or MAWs on trees~\cite{fici19minimal}.
Another extension is the computation within a sliding window~\cite{crochemore17minimal,mieno20minimal}, for which 
bounds on the number of changes in the answer set based on the sliding movement have been analyzed~\cite{akagi22combinatorics}.
	
\block{Anti-Power}
\cite{fici18antipower} introduced anti-powers as a new combinatorial concept.
We are aware of algorithms for computing anti-powers~\cite{badkobeh18antipower,alzamel19online},
and counting them~\cite{kociumaka22efficient}.
Anti-powers have been explicitly analyzed for prefixes of the Thue--Morse word~\cite{defant17antipower,gaetz21antipower}
and for aperiodic recurrent words~\cite{berger20antipowers}.
A specialization of anti-powers is to strengthen inequality by \emph{Abelian inequality}, where all factors of an Abelian anti-power must have pairwise distinct Parikh vectors~\cite{fici19antipower}.
Another specialization is due to~\cite{burcroff18klambda} with the restriction that the computed anti-power must be also a repetition, i.e., an integer power of a word.

\block{Longest-Previous-Factor}
The longest previous factor (LPF) array~\cite{franek03lpf,crochemore08lpf} has been proposed for pattern matching, computing the Lempel--Ziv 77 factorization, and detecting runs.
There are algorithms~\cite{crochemore08lpf,crochemore09lpf,chairungsee19approach} computing the LPF array in linear time.

\section{Preliminaries}
We start with introducing basic concepts for strings, and then formalize generalizations to indeterminate, GD, and ED strings.

\block{Strings} 
Let $\Sigma$ be an alphabet.
An element of $\Sigma^*$ is called a (classic) \emph{string}.
Given a string $T$, the $i$-th character of $T$ is denoted by $T[i]$, for an integer~$i \in [1..|T|]$,
where $|T|$ denotes the length of~$T$.
For two integers~$i$ and $j$ with $1 \leq i \leq j \leq |T|$,
a substring of $T$ starting at position $i$ and ending at position $j$ is
denoted by $T[i..j]$, i.e., $T[i..j] = T[i]T[i+1]\cdots T[j]$.
A substring $P$ of $T$ is called \emph{proper} if $P \neq T$.

\block{Indeterminate and Generalized/Elastic-Degenerated Strings}
We study the following extensions of classic strings.
For that, we make the distinction between \emph{characters} of the alphabet $\Sigma$ and \emph{symbols} of a string belonging to one of the three types of generalizations defined as follows.
First, an \emph{indeterminate string} is a string $\eds{S}[1..n]$ whose symbols are drawn from non-empty subsets of $\Sigma$, i.e., $\emptyset \neq \eds{S}[i] \subset \Sigma$.
Next, a \emph{generalized degenerate string} or \emph{GD string} is a string $\eds{S}[1..n]$ whose symbol at the $i$-th position is a non-empty subset of strings in $\Sigma^{k_i}$, 
i.e., $\emptyset \neq \eds{S}[i] \subset \Sigma^{k_i}$ for some $k_i$ depending on $i$.
Finally, an \emph{elastic-degenerate string} or \emph{ED string} is the most general type allowing each symbol to be a set of strings including the empty word~$\epsilon$, which is non-empty and not $\{\epsilon\}$, i.e., $S[i] \subset \Sigma^*$ with $\emptyset \neq S[i] \neq \{ \epsilon \}$.
Each indeterminate string can be expressed as a degenerate string by interpreting characters as strings of length one.
For any type of string class, we call its elements \emph{symbols}.
Like for strings, we denote with $\eds{S}[i..j] = \eds{S}[i] \cdots \eds{S}[j]$ the indeterminate/GD substring of that starts at position $i$ and ends at position $j$.

The \emph{Cartesian concatenation} of two symbols $X$ and $Y$ is $X \otimes Y := \{ xy \mid x \in X, y \in Y \}$.
Consider an ED string $\eds{S}$ of length $n$. 
The \emph{language} of an indeterminate/GD string $\eds{S}$ is 
$\Language(\eds{S}) := \Language(\eds{S})[1] \otimes \Language(\eds{S})[2] \otimes \cdots \otimes \Language(\eds{S})[n]$, cf.~\cite[Def.~5]{alzamel18degenerate}.
An element of $\Language(\eds{S})$ has length $n$ if $\eds{S}$ is indeterminate, or length $\sum_{i=1}^n k_i$ if $\eds{S}$ is a GD string.
To address a character in an ED string $\eds{S}$, we assume that each symbol $\eds{S}[i]$ is an ordered list 
(e.g., ordered canonically in lexicographic order) such that $\eds{S}[i][j][k]$ addresses the $k$-th character of the $j$-th string in the $i$-th symbol. We say that $(i,j,k)$ is the \emph{text-position} of $\eds{S}[i][j][k]$.
Regarding the ED strings of \cref{figExampleEDString}, $\eds{T_1}[1][1][1] = \texttt{A}, \eds{T_1}[2][2][1] = \texttt{C}, \eds{T_3}[4][3][3] = \texttt{T}$.
The \emph{size} of a symbol is the total length of the characters it contains, where the empty string is counted as having length one,
e.g., the sizes of the symbols of $\eds{T_3}$ are 1,2,1,3.
The \emph{size} of an ED string~$\eds{S}$, denoted by $||\eds{S}||$ is the sum of the sizes of all its symbols.
Finally, let $r(\eds{S})$ or shortly $r$ denote the maximum number of strings a symbol in $\eds{S}$ contains.

We say that a classic string $P$ has an \emph{occurrence} starting at text-position $(i,j,k)$ in an ED string $\eds{S}$ if 
we can factorize $P$ into $P = P_1 \cdots P_m$ such that $\eds{S}[i][j][k..] = P_1$, $P_{m}$ is a prefix of some string in $\eds{S}[i+m-1]$ and $P_x \in \eds{S}[i+x-1]$ for all $x \in [2..m-1]$.
For instance, $P = \texttt{CAC}$ has an occurrence in $\eds{T_2}$ at text-position $(2,2,1)$.

\section{Longest Repeating Substring}\label{secLRF}
The \emph{longest repeating substring (LRF)} of a classic string $T$ is the longest substring that occurs at least twice in $T$. 
We generalize LRF to ED strings as follows.

\begin{problem}[\PbLongestRepeating{} problem]
    Given an ED string $\eds{S}$, the \PbLongestRepeating{} problem is to determine a longest string that occurs at least twice in the language of~$\eds{S}$ but has to start at least at two different text-positions in $\eds{S}$.
\end{problem}

Our idea to compute \PbLongestRepeating{} is to reduce the problem to longest common extension (LCE) queries.
The LCE query in a classic string $T$ asks for the longest common prefix of two suffixes $T[i..]$ and $T[j..]$.
In the ED string setting, an LCE query at two text-positions asks for the longest string that has an occurrence starting at both text-positions.

\newcommand*{\LinS}{\ensuremath{L}}

Now, to find an LRF, we compute the LCE of every pair of pairwise different text-positions, and report the longest LCE length we discovered.
We can compute all LCE lengths in $\mathcal O(N^2)$ time with dynamic programming over the linearized form~\cite{buchler23efficient} of $\eds{S}$, where $N = ||\eds{S}||$ is the size of $\eds{S}$.
To this end, we introduce the linearized form of an ED string, which maps an ED string to a classic string whose alphabet contains three additional special characters.
The linearized form of an ED symbol $\edc{W_1, W_2, \ldots, W_k}$ is $\texttt{(}W_1 \texttt{|} W_2 \texttt{|} \ldots \texttt{|} W_k\texttt{)} \in 
(\Sigma \cup \{\texttt{(},\texttt{)},\texttt{|}\})^*$.
The characters in $\{\texttt{(},\texttt{)},\texttt{|}\}$ are called \emph{special} and are assumed to be not in $\Sigma$.
The linearized form of $\edc{W_1}$ is just $W_1$.
Then the linearized form of an ED string $\eds{S}$, denoted by $\LinS$, is the concatenation of the linearized forms of its symbols.
For example, the linearized form of $\eds{S}=\texttt{b}
\edc{\texttt{a}\\\varepsilon}\texttt{c}\edc{\texttt{abc}\\\texttt{c}}\edc{\texttt{a}\\\texttt{b}}$ is 
$\LinS  = \texttt{b(a|)c(abc|c)(a|b)}$.
Since the number of characters to write at most doubles, $|\LinS|\in\mathcal{O}(N)$.
The next function \(
\mathrm{next}(i, c) = \min\{j \mid j > i, \LinS[j] = c\}\cup\{1 + |\LinS|\}
\)
gives the first text position of an occurrence of $c$ in $\LinS[i+1..]$, or otherwise the invalid position~$1+|L|$.
Let 
\[
    \mathcal{P}(i) := \{i+1\}\cup\{j+1 \mid i<j<\mathrm{next}(i,\texttt{')'}) \wedge \LinS[j] = \texttt{'|'}\} \text{~for~} \LinS[i] = \texttt{'('}
\]
denote all starting positions of non-special character strings within the parentheses starting at $\LinS[i]$ divided by \texttt{'|'} characters.
For instance, $\mathcal{P}(7) = \{8, 12 \}$ and $\mathcal{P}(2) = \{3, 4\}$ for our example.
Then the function $\adj$ below gives the set of starting positions of non-special characters we want to match in case one of the characters is special.

\begin{minipage}{0.5\linewidth}
\(
\adj(i) = \begin{cases}
    \{i\} & \text{~if $\LinS[i] \in \Sigma$,}\\
    \{\mathrm{next}(i, \texttt{')'}) + 1\} & \text{~if $\LinS[i] = \texttt{'|'}$,} \\
    \{i+1\} & \text{~if $\LinS[i] = \texttt{')'}$,} \\
    \mathcal{P}(i) & \text{~if $\LinS[i] = \texttt{'('}$.} \\
\end{cases}
\)
\end{minipage}
\begin{minipage}{0.45\linewidth}
In particular, $\adj(i)$ advances if and only if $i$ is a position of a special character.
We can precompute $\adj(i)$ for all $i$ in linear time by scanning $\LinS$ from right to left with a stack maintaining the positions of the recently read special characters.
\end{minipage}
Let $D(i,j)$ be the longest common extension of two text-positions in \eds{S} that are mapped to the indices $i$ and $j$ in $\LinS$ (the mapping of text-positions to $\LinS$-positions is injective but not surjective).
We obtain $D$ by the following recursion formula.
\[
D(i,j) = \begin{cases}
    0 & \text{~if $\max\{i,j\} > |\LinS|$,} \\
    1 + D(i + 1,j + 1) & \text{~if $\LinS[i]=\LinS[j]$ and $\LinS[i]\in\Sigma$,} \\
    0 & \text{~if $\LinS[i]\neq\LinS[j]$ and $\LinS[i],\LinS[j] \in\Sigma$,} \\
    \max\{D(i',j') \mid (i',j')\in\adj(i)\times\adj(j)\} & \text{~otherwise.}
\end{cases}
\]

\begin{example}
We evaluate $D(1,9)$ of the above example. We visualize $\LinS$ with indices.

{\setlength{\tabcolsep}{2pt}
\begin{tabular}{l*{18}{c}}
    \hline
$i$ & 1 & 2 & 3 & 4 & 5 & 6 & 7 & 8 & 9 & 10 & 11 & 12 & 13 & 14 & 15 & 16 & 17 & 18 \\
$\LinS[i] = $ & \texttt{b} & \texttt{(} & \texttt{a} & \texttt{|} & \texttt{)} & \texttt{c} & \texttt{(} & \texttt{a} & \texttt{b} & \texttt{c} & \texttt{|} & \texttt{c} & \texttt{)} & \texttt{(} & \texttt{a} & \texttt{|} & \texttt{b} & \texttt{)} \\
    \hline
\end{tabular}
}
\begin{minipage}{0.35\linewidth}
To calculate $D(1,9)$, we recursively apply the above formula.
\end{minipage}

\begin{multicols}{2}
\begin{itemize}
    \item $D(1,9) = 1 + D(2,10)$
    \item $D(2,10) = \max\{ D(3,10), D(6,10) \}$ \\ $= D(6,10)$ because $D(3,10) = 0$
    \item $D(6,10) = 1 + D(7,11)$
    \item $D(7,11) = \max\{ D(8,14), D(12,14) \}$
    \item $D(8,14) = \max\{ D(8,15), D(8,17) \} = 1$
    \item $D(12,14) = \max\{ D(12,15), D(12,17) \} = 0$
\end{itemize}
\end{multicols}
We used that
$D(12,15) = 0$,
$D(12,17) = 0$,
$D(8,15) = 1$, and
$D(8,17) = 0$.
Hence, $D(1,9) = 1 + D(2,10) = 1 + D(6,10) = 1 + 1 + D(7,11) = 2 + D(8,15) = 3$.
\end{example}

We can compute $D(i,j)$ in constant time when neither $\LinS[i]$ nor $\LinS[j]$ is an open parenthesis character \texttt{'('}.
Because parentheses are neither nested nor overlapping in $\LinS$ by construction,
$\adj(i)\cap\adj(j) = \emptyset$ for all $i\neq j$ with $\LinS[i]=\LinS[j]=\texttt{'('}$.
Therefore, each pair of positions $(i',j')$ is included in at most one computation of $D(i,j)$ with $\LinS[i]=\LinS[j]=\texttt{'('}$.
Similarly, we use each $\adj(i)$ (with $\LinS[i]=\texttt{'('}$) in at most $\mathcal O(N)$ computations of $D(i,j)$ with $\LinS[j]\neq\texttt{'('}$.
By this amortization argument, we compute $D$ in $\mathcal O(N^2)$ time.

\begin{theorem}
    We can solve \PbLongestRepeating{} in $\mathcal O(N^2)$ time.
    \label{thm_longest_repeating_substring}
\end{theorem}

It is also unlikely that we can do better in general.
To understand that, we realize that $D$ can be used to solve the ED String Intersection problem (\PbStringIntersection{})~\cite{gabory23comparing}: Let $\eds{S_1}$ and $\eds{S_2}$ be two ED strings and $c \in \Sigma$ a character that has no occurrence in neither $\eds{S_1}$ nor $\eds{S_2}$.
Then the longest repeating substring of $\edc{c^k}\eds{S_1}\edc{c^k}\eds{S_2}\edc{c^k}$ has length at least $2k$ if and only if $\Language(\eds{S_1})\cap\Language(\eds{S_2})\neq\emptyset$ for a sufficiently large $k\in\mathcal O(N)$ with $N = ||\eds{S_1}|| + ||\eds{S_2}||$.
Because \cite[Thm.~2]{gabory23comparing} showed that there is no $\mathcal O(N^{2-\varepsilon})$-time algorithm for any constant $\varepsilon > 0$ for \PbStringIntersection{} under the Strong Exponential-Time Hypothesis (SETH), 
we conclude that there can be no $\mathcal O(N^{2-\varepsilon})$-time algorithm for \PbLongestRepeating{} under the same hypothesis.

\section{Hardness of Unique Substrings and Absent Words}\label{secMUSMAW}

In what follows, we reduce 3-SAT with $n$ variables to the decision problem whether a minimal unique substring (MUS) or a minimal absent word (MAW) of length at most $n$ exists.
The input of 3-SAT is a formula $F$ in conjunctive normal form (CNF), i.e., $F$ joins a set of clauses $C_i$ by conjunction (AND), where each clause $C_i$ is a disjunction of three literals.
We assume that the $n$ variables $x_1, \ldots, x_n$ are ranked.
The idea is to specify for each clause $C_i$ all variable assignments that do not make $C_i$ satisfiable. 
Such a set of assignments can be expressed by an indeterminate string~\eds{S} writing linearly all variables, where variables unused in $C_i$ can take any value.
By construction, subtracting the union of these unsatisfying assignments from all possible assignments gives us all satisfiable assignments.
The solution deviates from hereon for MUSs and MAWs:
For MUSs, if we also express the set of all assignments as an indeterminate substring of~\eds{S}, then the unsatisfiable assignments appear (at least) twice in \eds{S} while the satisfying assignments only once.
We show that these also coincide with the shortest MUSs if $F$ is satisfiable.
For MAWs, if we instead append an indeterminate string to~\eds{S} covering all possible substrings of length $n-1$, then a MAW must be of length at least $n$.
It is then left to show that a shortest MAW of~\eds{S} is substring of length $n$ that represents a satisfying assignment.

We start with a formal definition of MUSs based on classic strings.

\subsection{Hardness of Minimal Unique Substring}

We call a substring $U$ of a classic string $T$ \emph{unique} if it occurs exactly once in $T$, i.e., 
there is no other position in $T$ at which a substring starts that matches with $U$.
In particular, we call a unique substring $U$ to be a \emph{minimal unique substring} (\emph{MUS}) of $T$
if every proper substring of $U$ occurs at least twice in $T$.

We generalize the notion of MUSs to indeterminate strings by first simplifying the notion of an occurrence in indeterminate strings from ED strings.
For that, we say that a string $P \in \Sigma^*$ has an \emph{occurrence} in an indeterminate string $\eds{S}$ at position $i$ if 
there exists a string $X \in \Language(\eds{S})$ with $X[i..i+|P|-1] = P$.
We say a string $P$ in $\eds{S}$ is \emph{unique} if it has only one occurrence, meaning that the position of its occurrence is uniquely determined.
Further, We call $P$ a \emph{minimal unique substring (MUS)} if every proper substring $X$ of $P$ has at least two occurrences in $\eds{S}$.

We show that the decision version of finding a MUS in an indeterminate string is NP-hard.
With that we mean to find a unique substring of a specific length. 
The optimization version of minimizing the length then gives a MUS\@.
    
\begin{problem}[\PbUniqueSubstr{} problem]
    Given an indeterminate string $\eds{S}$ and an integer $k$, the \PbUniqueSubstr{} problem is to decide whether there is a unique substring of $\eds{S}$ with length at most $k$.
\end{problem}

\begin{theorem}
    \PbUniqueSubstr{} is NP-hard for $\sigma \ge 3$ and $r \ge 2$.
    \label{thm_mus_np_hard}
\end{theorem}
\begin{proof}
    Given a 3-CNF $F=C_1\land \dots\land C_m$ with $n$ variables and $m$ clauses. 
    We construct an indeterminate string $\eds{S}$ over the alphabet $\{0,1,\$\}$ such that there is a unique substring of $\eds{S}$ of length at most $n$ if and only if $F$ is satisfiable.
    We assume the variables are $x_1,\dots,x_n$.
    
    Let $\ell_a$, $\ell_b$, and $\ell_c$ be the literals of the clause $C_j = (\ell_a \lor \ell_b \lor \ell_c)$ such that $\ell_i$ is the literal of the variable $x_i$ for $i \in \{a,b,c\}$.
    For each $i$, we set $v_i = 0$ if $\ell_i$ is positive ($\ell_i = x_i$) and $v_i = 1$ otherwise ($\ell_i = \neg x_i$),
    and construct the following indeterminate string $\eds{T_j}$ based on the $v_i$'s.
    \[
        \eds{T_j} = \edc{0\\1}^{a-1} \edc{v_a} \edc{0\\1}^{b-a-1} \edc{v_b} \edc{0\\1}^{c-b-1} \edc{v_c} \edc{0\\1}^{n-c}.
    \]
    
    Then, $\eds{S}$ is given as follows:
    \[
        \eds{S} = \eds{T_1} \edc{\$} \dots \edc{\$} \eds{T_m} \edc{\$} \edc{0\\1}^n \edc{\$} \edc{0\\1}^{n-1} \edc{\$} \edc{0\\1}^{n-1}.
    \]
    $\eds{S}$ has $\mathcal O(nm)$ symbols because each $\eds{T_j}$ has $n$ symbols. 
    We make the following observations on $\eds{S}$:

    \begin{itemize}
        \item Each bitstring, i.e., a string of the alphabet $\{0,1\}$, of length $n$ occurs at least once, namely in the symbol $\edc{0\\1}^n$ at $\eds{S}[1 + (n+1)m]$.
        \item Each bitstring shorter than $n$ occurs more than once.
        \item No substring of length at most $n$ contains more than one `$\$$'.
        \item Every string of length at most $n$ containing a `$\$$' occurs more than once.
        \item A bitstring $B[1..n]$ of length $n$ encodes an assignment with $B[i] = 1$ if and only if $x_i$ is true.
            Such a bitstring $B$ occurs more than once in~\eds{S} if and only if its encoded assignment falsifies $F$. 
            Specifically, $B$ occurs in $\eds{T_i}$ if its assignment falsifies $C_i$.
    \end{itemize}
    Therefore, there can be no unique substring of length less than $n$, and a unique substring of length $n$ exists if and only if $F$ is satisfiable.
\end{proof}

\begin{example}\label{exMUS}
Given the 3-CNF $F=(x_1\lor\lnot x_2\lor x_3) \land (x_2\lor\lnot x_3\lor x_4)$ with $n = 4$ variables $x_1,x_2,x_3,x_4$.
The indeterminate string $\eds{S}$ for $F$
as described in the proof of Thm.~\ref{thm_mus_np_hard} is
\[
    \edc{0}\edc{1}\edc{0}\edc{0\\1}\edc{\$}\edc{0\\1}\edc{0}\edc{1}\edc{0}\edc{\$}\edc{0\\1}^4\edc{\$}\edc{0\\1}^3\edc{\$}\edc{0\\1}^3.
\]
The string $1000$ ($x_1=1, x_2=x_3=x_4=0$) is a MUS because it occurs only once in $\eds{S}$, all strings of length $3$ on the alphabet $\{0,1\}$ occur in $\eds{S}$.
Other strings of length 4 such as $0100$ and $00\$1$ occur twice.
Thus, $1000$ is one of the shortest MUS in $\eds{S}$.
\end{example}

\subsection{Hardness of Minimal Absent Word}
An \emph{absent word} of a classic string~$T$ is a non-empty string that is not a substring of $T$.
An absent word $X$ of $T$ is called a \emph{minimal absent word} (\emph{MAW}) of $T$ if all proper substrings of $X$ occur in $T$.
More general, we say that a string in an indeterminate string~\eds{S} is \emph{absent} if it does not occur in~\eds{S}.
An absent string~$X$ is \emph{minimal} in~\eds{S} if every proper substring of $X$ occurs in~\eds{S}.
In what follows, we show that the decision problem for finding an absent word of at most a specific length is NP-hard.

\begin{problem}[\PbAbsentWord{} problem]
    Given an indeterminate string $\eds{S}$ and an integer $k$, the $k$-\PbAbsentWord{} problem is to decide whether there is an absent word of $\eds{S}$ with length at most $k$.
\end{problem}

\begin{theorem}
    \PbAbsentWord{} is NP-hard for $\sigma \ge 3$ and $r \ge 3$.
    \label{thm_maw_np_hard}
\end{theorem}
\begin{proof}
Our idea is to follow the proof of \cref{thm_mus_np_hard}, except that we define $\eds{S}$ as
    \[
        \eds{S} = \eds{T_1} \edc{\$} \dots \edc{\$} \eds{T_m} \edc{\$}  \edc{0\\1\\\$}^{n-1} 
        \edc{\$} \edc{0\\1}^{n-1}.
    \]
$\eds{S}$ contains all strings from the alphabet $\{0,1,\$\}$ with a length of at most $n-1$ as a substring.
An absent word thus must be of length at least $n$.
Since any string of length at most $n$ containing a $\$$ occurs also in the last three symbols of $\eds{S}$,
the search for an absent word of length at most $n$ boils down to finding a bitstring of length $n$.
Each bitstring~$B$ encodes an assignment. 
If the assignment does not satisfy clause $C_i$, then $B$ matches $\eds{T_i}$.
On the one hand, if there is no satisfying assignment for the CNF, then all substrings of length $n$ with alphabet $\{0,1\}$ must be present in $\eds{S}$.
In total, there is no absent word of length at most $n$ in $\eds{S}$.
On the other hand, if there is a satisfying assignment, there is also an absent word of length $n$, which is a bitstring encoding this assignment. 
\end{proof}

\begin{example}
We reuse the 3-CNF $F$ from \cref{exMUS}, i.e.,
$F=(x_1\lor\lnot x_2\lor x_3) \land (x_2\lor\lnot x_3\lor x_4)$.
The indeterminate string $\eds{S}$ for $F$
as described above is
\[
    \eds{S} =
    \edc{0}\edc{1}\edc{0}\edc{0\\1}\edc{\$}\edc{0\\1}\edc{0}\edc{1}\edc{0}\edc{\$}
    \edc{0\\1\\\$}^3\edc{\$}\edc{0\\1}^3.
\]
The string $1000$ ($x_1=1, x_2=x_3=x_4=0$) does not occur in $\eds{S}$.
However its longest proper prefix $100$ and suffix $000$ occur each in $\eds{S}$.
Hence, $1000$ is a MAW\@.
It is also a shortest MAW since each string of length 3 appears in $\eds{S}$.
\end{example}

The \PbAbsentWord{} problem is FPT in the size $s$ of the language $\Language(\eds{S})$.
The algorithm of~\cite{okabe23lineartime} computes a generalized MAW of a set of strings $\mathcal S$ in $\mathcal O(m)$ time, where $m$ is the total length of all strings in $\mathcal S$.
Here, a \emph{generalized MAW} is a string that is a MAW for all strings in $\mathcal S$.
Hence, computing the generalized MAW of $\Language(\eds{S})$ takes $\mathcal O(n s)$ time, and the computed generalized MAW is a MAW of $\eds{S}$.

\subsection{SAT Formulation}
In what follows, we present  SAT formulations for computing a unique substring or an absent word of a specified length~$x$ in an indeterminate string.
The SAT formulations can be transformed into a MAX-SAT formulation by an optimizing objective on $x \in [1..n]$.
For the formulations, we represent a symbol of an indeterminate string $\eds{T}[1..n]$ as a list of characters such that $\eds{T}[i]$ is a list.
We let $\eds{T}[i][k]$ denote the $k$-th character in the list $\eds{T}[i]$ and $|\eds{T}[i]|$ the number of characters stored in $\eds{T}[i]$.
We interpret the Boolean variables true and false as integers 1 and 0 to allow us to use expressions such as summations.
We model our answer as a string $X[1..x] \in \Sigma^x$ of length $x$. 
The length $x$ can be given (decision problem) or is object to optimization
to obtain the shortest length. We can model $x$ by $n$ Boolean variables $x_i$ as follows.

\begin{minipage}{0.45\linewidth}
\begin{align}\label{eqLENX}
\sum_{i=1}^n x_i = 1.
\tag{LENX}
\end{align}
\end{minipage}
\begin{minipage}{0.45\linewidth}
\begin{align}\label{eqMINX}
    \min \{i \in [1..n] : x_i = 1 \}.
\tag{MINX}
\end{align}
\end{minipage}

In the following, variables that are not defined are assumed to be false (not set).

\subsubsection{Shortest Minimal Absent Word (MAW)}\label{secSATMAW}
\begin{minipage}{0.5\linewidth}
    $X$ can be expressed by $x\sigma$ Boolean variables arranged in a matrix $X'[1..x][1..\sigma]$ choosing characters from $\Sigma = \{1,\ldots,\sigma\}$.
\end{minipage}
\begin{minipage}{0.5\linewidth}
\begin{align}\label{eqSETX}
\forall \ell \in [1..x]: \sum_{c=1}^{\sigma} X'[\ell,c] = 1.
\tag{SETX}
\end{align}
\StateComplexity{$\mathcal O(x)$, $\mathcal O(\sigma)$}
\end{minipage}

The gray curly brackets denote two asymptotic upper bounds, the first is on the number of generated clauses and the second on the maximum size of each such clause.
Here, we generate for each $X$-position a clause of size $\sigma$.

To check whether $X$ is absent, we need to verify that there is no occurrence of $X$ starting at any text position $t \in [1..n-x+1]$ of $\eds{T}$.
In other words, we check for each such $t$ that there is an $X$-position $\ell \in [1..x]$ such that $X[\ell] \not\in \eds{T}[t+\ell-1]$.
To express that as a formula, we create a three-dimensional grid of boolean variables $M[k,t,\ell]$ modelling these mismatches.
We set $M[k,t,\ell]$ to true if $X[\ell] \neq \eds{T}[t+\ell-1][k]$, where $\eds{T}[t+\ell-1][k]$ denotes the $k$-th character in the list $\eds{T}[t+\ell-1]$, if it exists.
\begin{align}\label{eqM}
    \forall \ell \in [1..x], t \in [1..n-x+1], k \in [1..|\eds{T}[t+\ell-1]|]: X[\ell] \neq \eds{T}[t+\ell-1][k] \implies M[k,t,\ell].
    \tag{M}
\end{align}
\StateComplexity{$\mathcal O( x n r)$, $\mathcal O(1)$}
Recall that $r = r(\eds{T}) \le \sigma$ is the maximum number of characters a list $\eds{T}[i]$ can store.
Next we reduce $M[k,t,\ell]$ to $M'[t,\ell]$ by setting $M'[t,\ell]$ to true only if $X[\ell]$ mismatches with all characters in the list $\eds{T}[t+\ell-1]$, i.e.,
\begin{align}\label{eqM'}
    \forall \ell \in [1..x], t \in [1..n-x+1]:
    \sum_{k=1}^{|\eds{T}[t+\ell-1]|} M[k,t,\ell] = |\eds{T}[t+\ell-1]| \implies M'[t,\ell].
    \tag{M'}
\end{align}
\StateComplexity{$\mathcal O(xn)$, $\mathcal O(r)$}

So $M'[t,\ell]$ is true if and only if $X[\ell] \not\in \eds{T}[t+\ell-1]$.
Finally, we require that there is some $X$-position $\ell$ for which we cannot match $X[\ell]$ with any of the characters of $\eds{T}[t+\ell-1]$, i.e.,
\begin{align}\label{eqCONS}
    \forall t \in [1..n-x+1]: \sum_{\ell=1}^{x} M'[t,\ell] \neq 0.
    \tag{CONS}
\end{align}
\StateComplexity{$\mathcal O(n)$, $\mathcal O(x)$}
If Eq.~(\ref{eqCONS}) holds, then $X$ must be an absent word.
In total, we have $x\sigma + n$ selectable Boolean variables.
The number of clauses is in $\mathcal O(xnr)$.
The same upper bound has the size of the CNF being the summation of the sizes of all clauses.

Listing~\ref{lstMAW} shows an encoding in answer set programming (ASP), where each line is commented by one of the above equations.

\begin{lstlisting}[label=lstMAW, caption={ASP encoding computing a MAW. The last line (command \texttt{show}) is for the output.}]
1 { len(X) : X = 1..n } 1. %(@\ref{eqLENX}@)
#minimize { X : len(X) }. %(@\ref{eqMINX}@)
1 { x(L,C) : a(C) } 1 :- len(X), L=1..X. %(@\ref{eqSETX}@)
m(K,T,L) :- t(K,T+L-1,C), x(L,D), C!=D, len(X), T=1..n-X+1. %(@\ref{eqM}@)
m(T,L) :- r { m(K,T,L) : K=1..r }, len(X), L = 1..X, T=1..n-X+1. %(@\ref{eqM'}@)
:- { m(T,L) } 0, len(X), T=1..n-X+1. %(@\ref{eqCONS}@)
#show x/2.
\end{lstlisting}

\begin{table}[t]
    \begin{minipage}{0.5\linewidth}

\begin{tabular}{r*{6}{r}}
\toprule
		& \multicolumn{2}{c}{MAW} & \multicolumn{2}{c}{MUS} \\
\cmidrule{2-3} \cmidrule{4-5}
$n$ & time & $x$ & time & $x$
\\\midrule
\num{10} & \num{0.02} & \num{2} & \num{0.03} & \num{2} \\
\num{100} & \num{169.62} & \num{4} & \num{171.60} & \num{3} \\
\num{150} & \num{566.88} & \num{4} & \num{580.68} & \num{4} \\
\num{200} & \num{1365.87} & \num{4} & \num{1376.11} & \num{4} \\
\num{250} & \num{2615.42} & \num{5} & \num{2708.00} & \num{5} \\
 \end{tabular}

     \end{minipage}
    \hfill
    \begin{minipage}{0.45\linewidth}
        \caption{Computation of a shortest MAW and MUS on a randomly generated indeterminate string~\eds{S} with $\sigma = 4$ and $r(\eds{S}) = 2$. 
        Time is in seconds, and $n$ is the length of the prefix~\eds{T} of \eds{S} taken as input for computation. }
        \label{tabASP}
    \end{minipage}
\end{table}

\subsubsection{Shortest Minimal Unique Substring (MUS)}
To find a unique substring of length $x \in [1..n]$, we slightly modify the above SAT formulation by representing $X$ differently as follows.
We first select a starting position $p \in [1..n]$ of $X$ with $n$ Boolean variables $p_i$ with only one set that marks the start of $X$ in~\eds{T}, i.e.,

\begin{align}\label{eqPOSP}
\sum_{i=1}^n p_i = 1.
\tag{POSP}
\end{align}

Each character of $X$ can be chosen such that $X[\ell] \in \eds{T}[p+\ell-1]$ for every $\ell \in [1..x]$.
For that we create a two-dimensional table $X'[\ell,k]$ such that $X'[\ell,k] = 1$ if and only if $X[\ell] = \eds{T}[p+\ell-1][k]$.
This can be modelled by the following constraint.
\begin{align}\label{eqSELX}
    \forall \ell \in [1..x]: \sum_{k=1}^{|\eds{T}[p+\ell-1]|} X'[\ell,k] = 1.
\tag{SELX}
\end{align}
\StateComplexity{$\mathcal O(x)$, $\mathcal O(r)$}

Then we define $M$ and $M'$ as above but omit in Eq.~(\ref{eqCONS}) the starting position $p$ of $X$ in the range for $t \in [1..n-x+1] \setminus \{p\}$.
The asymptotic complexity of the generated CNF compared to \cref{secSATMAW} is that the number of selectable variables is $2n + rx$, which can be fewer for large $x \in \Theta(n)$ and small $r$.
Listing~\ref{lstMUS} shows an encoding in ASP\@.
An evaluation on a randomly generated string is shown in \cref{tabASP}.
The evaluation ran on an Intel Xeon Gold 6330 CPU with Ubuntu 22.04 and clingo version 5.7.1 for interpreting the ASP encoding.

\begin{lstlisting}[label=lstMUS, caption={ASP encoding computing a MUS.}]
1 { len(X) : X = 1..n } 1. %(@\ref{eqLENX}@)
#minimize { X : len(X) }. %(@\ref{eqMINX}@)
1 { pos(P) : P = 1..n } 1. %(@\ref{eqPOSP}@)
1 { x(L,C) : t(K,P+L-1,C) } 1 :- len(X), pos(P), L = 1..X. %(@\ref{eqSELX}@)
m(K,T,L) :- t(K,T+L-1,C), x(L,D), C != D, len(X), T=1..n-X+1. %(@\ref{eqM}@)
m(T,L) :- r { m(K,T,L) : K=1..r }, len(X), L = 1..X, T=1..n-X+1. %(@\ref{eqM'}@)
:- { m(T,L) } 0, len(X), pos(P), T=1..n-X+1, T != P.  %(@\ref{eqCONS}@)
#show x/2.
\end{lstlisting}

\section{Hardness of $k$-\antipower}
A factorization of a classic string $T$ is a partition of $T$ into factors $F_1, \ldots, F_k$ such that their concatenation gives $T$, i.e., $F_1 \cdots F_k = T$.
$T$ has an \emph{anti-power} of order $k$ if $T$ can be factorized into $k$ factors of equal length such that all factors are distinct.
As a generalization, we say that a GD string has a \emph{$k$-anti-power} if its language contains a $k$-anti-power.

\begin{problem}[\antipower{} problem]\label{pbAntiPower}
Given a GD string $\eds{S}$ and an integer $k$, the \antipower{} problem is to decide whether $\eds{S}$ has an anti-power of order $k$.
\end{problem}

If the cardinality of $\Language(\eds{S})$ is polynomial in the length $m (= \sum_{i=1}^{n} |\eds{S}[i][1]|)$ of a string in $\Language(\eds{S})$, 
then we can try to compute a $k$-anti-power in $\mathcal O(m^2/k)$ time with the algorithm of~\cite{badkobeh18antipower} for every string in $\Language(\eds{S})$.
This gives an algorithm for $k$-\antipower{} whose time is polynomial in $m$.
However, the cardinality of $\Language(\eds{S})$ can be exponential.

\begin{theorem}
    \antipower is NP-hard for $\sigma \ge 2$ and $r \ge n$.
    \label{thm_antipower_np_hard}
\end{theorem}
\begin{proof}
    We reduce from \PbHamiltonianPath{}, which belongs to one of Karp's classic NP-hard problems~\cite{karp72reducibility}.
    Given an undirected graph $G=(V,E)$ with $n=|V|\geq 2$ without self-loops.
    We construct a GD string $\eds{S}$ such that $\eds{S}$ has an anti-power of order $k = 2(\binom{n}{2}-|E|+n)-1$ if and only if $G$ has a Hamiltonian path.
    For this, we assume that we can enumerate the vertices such that $V=\{1,\dots,n\}$.
    Then we can identify each vertex $v_i \in V$ with a character $c_{v_i}$ of an alphabet $\Sigma=\{c_1,\dots,c_n\}$.
    Let $\overline{E}=\binom{V}{2}\setminus E=\{e_1,\dots,e_{|\overline E|}\}$ be the complement of $E$, again without self-loops.
    For each $e_i=\{u_i,v_i\} \in \overline{E}$ we define the string 
    $S_i = c_{u_i}c_{v_i} \cdot c_{v_i}c_{u_i} \in \Sigma^4$.
    Then define the GD string $\eds{S}$ as
    \[
        \eds{S} = S_1 \cdots S_{|\overline{E}|}\edc{c_1^3\\\vdots\\c_n^3}\edc{c_1^4\\\vdots\\c_n^4}^{n-2}\edc{c_1^3\\\vdots\\c_n^3}.
        \text{\hspace{1em}We make the following observations on $\eds{S}$.}
    \]
    \begin{itemize}
        \item 
            The length of each string of the language $\Language(\eds{S})$ is $4(\binom{n}{2}-|E|+n)-2$ due to the following two reasons, where we split the analysis by the prefix of the $S_i$'s and the remaining suffix.
            First, the cardinality of $\overline{E}$ is $\binom{n}{2} - |E|$ and therefore the summation of the lengths of all $S_i$'s accumulates to $4|\overline{E}| = 4(\binom{n}{2}-|E|)$. 
            The remaining suffix has length $3+4(n-2)+3 = 4n - 2$.
            The summation of both is $4(\binom{n}{2}-|E|+n) - 2$, which is $2k$.
        \item A $k$-anti-power $X$ of $\eds{S}$ is a member of $\Language(\eds{S})$, and thus must have length $2k$. 
            In particular, $X$ consists of $k$ factors of equal length $2$.
            By construction, each factor of $X$ starts at an odd position in \eds{S}.
\item Every string $Y \in \Language(\eds{S})$ corresponds to a walk~$W[1..n]$ over $n$ nodes in the complete graph $(V, \binom{V}{2})$ without self-loops, 
            and vice versa:
To understand that, let us focus on the suffix $\eds{S}[4|\overline E|+1..]$ directly after the sequence of $S_i$'s.
    For each $j \in [1..n]$, let $s_j=1 + 4|\overline E| + 4(j-1)$. 
    Then $Y[s_j..s_j+1]$ is of the form $c_ic_i$ for some $i \in [1..n]$ and encodes that the $j$-th visited node of $W$ is node $i$, i.e., $W[j] = i$.
    By construction for $j < n$, $Y[s_j+2..s_j+3]$ is of the form $c_ac_b$ where $W[j] = a$ and $W[j+1] = b$ are the $j$-th and $(j+1)$-st node on $W$, respectively.
    Now take a walk~$W[1..n]$ over $n$ nodes represented by a string $Y$.
        $W$ is a permutation if and only if it is Hamiltonian (pigeonhole principle).
        \item Iff $W$ is not Hamiltonian, $W$ visits at least one node $i$ more than once.
            In this case, $c_ic_i$ occurs twice in~$Y$, starting at odd positions. Therefore, $Y$ cannot be a $k$-anti-power.
        \item Iff $W$ uses an edge from $u$ to $v$ with $u \neq v$, there is some $j$ such that $Y[s_j..s_j+5] = c_uc_uc_uc_vc_vc_v$.
            Assume that this edge $\{u,v\}$ is not in $E$.
            Since $\{u,v\}\in\overline E$, $c_uc_v$ also occurs starting at some odd $p< 4|\overline E|$ (i.e.\ in the part of $Y$ generated by an $S_i$).
            Therefore, $Y$ cannot be a $k$-anti-power.
        \item Because of the previous two statements, a Hamiltonian path in $(V,E)$ corresponds to a $k$-anti-power. Since every walk of length~$n$ (and thus every Hamiltonian path) corresponds to a string in $\Language(\eds{S})$, the converse is also true.
    \end{itemize}
    Finally, we can generalize the proof to GD strings over the binary alphabet $\{0,1\}$ by replacing each $c_i$ with a unique bitstring ($\in \{0,1\}^*$) of length $\lceil\log_2 n\rceil$.
    The anti-power we aim to compute is still of order~$k$, but each factor then contains $\lceil\log_2 n\rceil$ bits.
\end{proof}

\begin{example}
    Consider the following graph, with characters 
    $V = \{ \texttt{a}, \texttt{b}, \texttt{c}, \texttt{d} \}$ 
    instead of numbers as node labels for readability.

    \begin{minipage}{0.2\linewidth}
    \tikzset{every picture/.style={/utils/exec={\ttfamily}}}
\begin{tikzpicture}
    \graph { a -- {b, c -- d}, b -- c};
\end{tikzpicture}
    \end{minipage}
    \hfill
    \begin{minipage}{0.75\linewidth}
The set of edges is $E = \{\{\texttt{a},\texttt{b}\},\{\texttt{a},\texttt{c}\}, \{\texttt{b},\texttt{c}\}, \{\texttt{c},\texttt{d}\}\}$.
Its complement, omitting self-loops, is
$\overline E= \{\{\texttt{a},\texttt{d}\},\{\texttt{b},\texttt{d}\}\}$.
The corresponding GD string is
\newcommand{\aedc}[1]{\texttt{\textcolor{solarizedBlue}{#1}\textcolor{solarizedYellow}{#1#1}\textcolor{solarizedBlue}{#1}}}
\newcommand{\bedc}[1]{\texttt{\textcolor{solarizedYellow}{#1#1}\textcolor{solarizedBlue}{#1}}}
\newcommand{\cedc}[1]{\texttt{\textcolor{solarizedBlue}{#1}\textcolor{solarizedYellow}{#1#1}}}
\newcommand{\eedc}[1]{\edc{\texttt{\textcolor{solarizedBlue}{#1}}}}
\[
    \eds{S} = \eedc{ad}\eedc{da}\eedc{bd}\eedc{db}\edc{\bedc{a}\\\bedc{b}\\\bedc{c}\\\bedc{d}}\edc{\aedc{a}\\\aedc{b}\\\aedc{c}\\\aedc{d}}\edc{\aedc{a}\\\aedc{b}\\\aedc{c}\\\aedc{d}}\edc{\cedc{a}\\\cedc{b}\\\cedc{c}\\\cedc{d}}.
\]
\end{minipage}

There are four anti-powers of $\eds{S}$ of order $k=11$, namely

\begin{itemize}
\item 
\texttt{\textcolor{solarizedBlue}{addabddb}\textcolor{solarizedYellow}{aa}\textcolor{solarizedBlue}{ab}\textcolor{solarizedYellow}{bb}\textcolor{solarizedBlue}{bc}\textcolor{solarizedYellow}{cc}\textcolor{solarizedBlue}{cd}\textcolor{solarizedYellow}{dd}}
for $\texttt{a} \to \texttt{b} \to \texttt{c} \to \texttt{d}$,
\item 
\texttt{\textcolor{solarizedBlue}{addabddb}\textcolor{solarizedYellow}{bb}\textcolor{solarizedBlue}{ba}\textcolor{solarizedYellow}{aa}\textcolor{solarizedBlue}{ac}\textcolor{solarizedYellow}{cc}\textcolor{solarizedBlue}{cd}\textcolor{solarizedYellow}{dd}}
for $\texttt{b} \to \texttt{a} \to \texttt{c} \to \texttt{d}$,
\item 
\texttt{\textcolor{solarizedBlue}{addabddb}\textcolor{solarizedYellow}{dd}\textcolor{solarizedBlue}{dc}\textcolor{solarizedYellow}{cc}\textcolor{solarizedBlue}{ca}\textcolor{solarizedYellow}{aa}\textcolor{solarizedBlue}{ab}\textcolor{solarizedYellow}{bb}}
for $\texttt{d} \to \texttt{c} \to \texttt{a} \to \texttt{b}$, and
\item 
\texttt{\textcolor{solarizedBlue}{addabddb}\textcolor{solarizedYellow}{dd}\textcolor{solarizedBlue}{dc}\textcolor{solarizedYellow}{cc}\textcolor{solarizedBlue}{cb}\textcolor{solarizedYellow}{bb}\textcolor{solarizedBlue}{ba}\textcolor{solarizedYellow}{aa}}
for $\texttt{d} \to \texttt{c} \to \texttt{b} \to \texttt{a}$.
\end{itemize}
From each anti-power we can construct a Hamiltonian path by reading each same-character pair in yellow color
({\color{solarizedYellow}$\blacksquare$})
.
\end{example}

\section{Hardness of Longest Previous Factor}

The \PbLongestPrevious{} problem on a classic string $T[1..n]$ is to compute, 
for a given position $T[i]$, the length of a substring starting before $i$ that shares the longest common prefix with $T[i..]$, i.e.,
\[
    \mathsf{LPF}[i] := \max \{ \ell \in [0..n-i+1]  \mid T[j..j+\ell-1] = T[i..i+\ell-1] \wedge j < i\}.
\]
$\mathsf{LPF}[i]$ is the length of the \emph{longest previous factor (LPF)} of $T[i]$.
We can translate this problem to ED strings by a weak and a strong variant.
Given a text-position $(i,j,k)$ in an ED string~\eds{S}, 
the weak variant states that an LPF is the longest string that has an occurrence at $(i,j,k)$ and an earlier occurrence at $(i',j',k') < (i,j,k)$,
where $(i',j',k') < (i,j,k)$ if $i' < i$ or ($i' = i$, $j' = j$, and $k' < k$).
Informally, the strong variant requires the LPF to choose the same path through~\eds{S} in case it takes characters from $\eds{S}[i]$.
More technically, 
the LPF of $(i,j,k)$ is the LPF at position $x$ of a classic string $X$ in the language~$\Language(\eds{S})$ such that $X$ covers $(i,j,k)$ at position $x$.
With covering we mean that there is a factorization $X = X_1 \cdots X_n$ such that each $X_z$ matches with $S_z$ and $X[i] = \eds{S}[i][j]$ with $|X_1 \cdots X_{i-1}| + k = x$.
Consider for example
\(
\eds{S} = 
\edc{\texttt{ab}}\edc{\texttt{abc}\\\texttt{c}}.
\)
In the strong variant, the LPF of text-position $(2,1,1)$ is $\texttt{ab}$ because there is only one string $X = \texttt{ababc} \in \Language(\eds{S})$ covering $(2,1,1)$ at position~$3$, and the LPF of $X[3]$ is $\texttt{ab}$.
In the weak variant, $\texttt{abc}$ is also permissible.
The weak variant can be reduced to computing LCEs, and therefore is solvable in quadratic time (cf.~\cref{secLRF}).
The variants only emerge for ED strings; computing the LPF in a GD string can be treated as in the weak variant.
This strong variant is at least as hard as computing the maximum length~$\ell$ of the LPFs of the strings in $\Language(\eds{S})$, i.e.,\ 
\(
    \ell = \max\{\max\{\mathsf{LPF}(T)[i] \mid 1\leq i\leq |T|\} \mid T \in\Language(\eds{S})\}.
\)
We now show that computing $\ell$ is NP-hard, formalized by the following decision problem:

\begin{problem}[\PbLongestPrevious{} problem]
    Given an ED string $\eds{S}$ and an integer $k$, the \PbLongestPrevious{} problem is to decide whether the longest LPF in $\Language(\eds{S})$ is at least as long as $k$.
\end{problem}

Like \cref{pbAntiPower}, \PbLongestPrevious{} can be answered in polynomial time if $\Language(\eds{S})$ has polynomial cardinality
because we can compute the LPF array for each string in $\Language(\eds{S})$ in linear time to its length~\cite{crochemore08lpf}.

\begin{theorem}
    \PbLongestPrevious{} is NP-hard for $\sigma \ge 3$ and $r \ge \sigma+1$.
    \label{thm_lpf_np_hard}
\end{theorem}
\begin{proof}
    We reduce from \PbCommonSubsequence{}.
    Let $\mathcal S$ be a set of strings and $k\geq 1$ an integer.
    \PbCommonSubsequence{} asks whether there is a string of length $k$ that is a subsequence common to all strings in $\mathcal S$. 
    \PbCommonSubsequence{} is NP-hard~\cite{maier78complexity} in the cardinality of $\mathcal S$ for alphabet sizes $\sigma \ge 2$.

    The general idea is to construct an ED string such that every sufficiently long repeated factor~$F$ overlaps and therefore is periodic, 
    such that $F$'s root is a subsequence of length $k$ (modulo a separator symbol) of some or all strings in $\mathcal S$.
    Here, a string~$S$ is called \emph{periodic} if the largest possible rational number $k \ge 1$ with $S = R^k$ for a string $R$ is at least two; in that case we call $R$ the \emph{root} of $S$.
    In what follows, we show that the longest such factor~$F$ is a subsequence of length $k$ common to \emph{all} strings in $\mathcal S$.

    Assume that all strings in $\mathcal S = \{S_1,\dots, S_f\}$ are over the alphabet $\Sigma$. With a slight abuse of notation we consider $\Sigma$ as the ED symbol $\{ x \mid x \in \Sigma\} \cup \{\epsilon\}$ matching any character of $\Sigma$ or nothing.
    We now construct an ED string $\eds{S}$ that has a longest LPF of length $(2f+1)(k+1)+1$ if and only if $\mathcal S$ has a common subsequence of length $k$.
    For each $i \in [1..f]$ we define 
    \(
        \eds{S_i} = \edc{S_i[1]\\\varepsilon}\edc{S_i[2]\\\varepsilon}\cdots\edc{S_i[|S_i|]\\\varepsilon}.
    \)
    By construction, $\Language(\eds{S_i})$ is the set of subsequences of $S_i$.
    Let $\$$ be a symbol not in $\Sigma$.
    Let $\ell\geq 3$ be minimal such that $k\cdot(\ell-2) > \sum_i|S_i| = ||\mathcal S||$, where $||\mathcal S||$ denotes the accumulated lengths of all strings in $\mathcal S$.
    Define 
    \[
        \eds{S} = \left(\edc\$\Sigma^k\right)^\ell\edc{\$}\eds{S_1}\edc{\$}\cdots\edc\$\eds{S_f}\edc{\$}\Sigma^k\edc\$.
    \]
    By the minimality of the chosen value of $\ell$, $||\eds{S}|| \in \mathcal O(||\mathcal S||)$.
    Now consider the LPF of the position of the second occurrence of $\edc{\$}$ in $\eds{S}$, i.e., symbol $\eds{S}[k+2]$.

    \begin{itemize}
        \item The only non-empty previous factor $F$ of $\eds{S}[k+2]$ starts at the very beginning of $\eds{S}$ because it must start with '$\$$'. Assume that $F$ is the longest possible such factor.
        \item On the one hand, $F$ has length at least $(\ell-1)(k+1)+1$ due to the prefix $\left(\edc\$\Sigma^k\right)^\ell$ of $\eds{S}$ with $\ell\geq 3$ matching the first $(\ell-1)(k+1)+1$ symbols starting at $\eds{S}[k+2]$.
        \item On the other hand, $F$ is of the form $(\$W)^z\$(W[1..p])$ for some $z\geq0$, $p \in [0..k-1]$, $W\in\Sigma^k$ because the occurrences of $F$ starting at~$\eds{S}[1]$ and at the second $\edc{\$}$ overlap.
    \end{itemize}
        If $|F| = (k+1)(f+\ell)+1$, $W$ is a common subsequence of all strings in $\mathcal S$. 
            That is because $F$ is a prefix of a string $X \in \Language(\eds{S})$ by the definition of the strong variant of \PbLongestPrevious{},
            and therefore its substrings $Y_i$ and $Y_{i+1}$ matching with $\{\$\}\eds{S_i}\{\$\}$ and $\{\$\}\eds{S_{i+1}}\{\$\}$, respectively, must match, i.e., $Y_i = Y_{i+1}$ for every $i \in [1..f-1]$.
            Thus, $W = \$ Y_i \$$ is determined by $\eds{S_i}$, for all $i$.
        Finally, $F$ cannot be longer than $(k+1)(f+\ell)+1$ 
            due to the number of $\$$'s and since every $(k+1)$-st character of $F$ is $\$$.
    It is left to show that all other text-positions have a shorter LPF.
            \begin{itemize}
                \item An occurrence of a string starting after the prefix $\left(\edc\$\Sigma^k\right)^{\ell - 1}$ of \eds{S} can have length at most 
                    \begin{align*}
                    3 + 2k + \sum_{i=1}^f {\left(1 + |S_i|\right)} 
                        =\ & 3 + 2k + f + \sum_{i=1}^f {|S_i|} \\
                    <\ & 3 + 2k + f + k(\ell-2) 
                    =\  3 + f + k\ell \\
                    <\ & kf + f + k\ell+\ell+1 
                    =\  (k+1)(f+\ell) + 1
                \end{align*}
                \item An occurrence of a string starting at a text-position $p$ before $\left(\edc\$\Sigma^k\right)^{\ell-1}$ with length of at least $2k+2$ has a substring matching $\edc{\$}\Sigma^k\edc{\$}$.
                A sufficiently long LPF for text-position $p$ must overlap by construction, and thus again every $(k+1)$-st character must be $\$$.
                When $p$ is after the second $\$$, the LPF of $P$ is shorter than $(k+1)(f+\ell)+1$.
            \end{itemize}
\end{proof}

\begin{example}
    As an example for the proof of \cref{thm_lpf_np_hard} consider $\mathcal S = \{S_1, S_2\}$ with $f = |\mathcal S| = 2$, $S_1 = \texttt{abb}$ and $S_2 = \texttt{bab}$.
    We want to test whether $\mathcal S$ has a common subsequence of length $k=2$. 
    Therefore, we need to select the smallest $\ell > (|S_1| + |S_2|)/k + 2 =  5$, which is $\ell = 6$.
    Then 
$\eds{S_1} = \edc{\texttt{a} \\ \epsilon} \edc{\texttt{b} \\ \epsilon} \edc{\texttt{b} \\ \epsilon}$,
$\eds{S_2} = \edc{\texttt{b} \\ \epsilon} \edc{\texttt{a} \\ \epsilon} \edc{\texttt{b} \\ \epsilon}$, and
\begin{align*}
    \eds{S} &= \left( \$ \edc{\texttt{a} \\ \texttt{b} \\ \epsilon}^2 \right)^6 \$ \eds{S_1} \$ \eds{S_2} \$ \edc{\texttt{a} \\ \texttt{b} \\ \epsilon}^2 \$  \\
         \\ &= \$ \edc{\texttt{a} \\ \texttt{b} \\ \epsilon}^2  \$ \left( \edc{\texttt{a} \\ \texttt{b} \\ \epsilon}^2 \$ \right)^5 \edc{\texttt{a} \\ \epsilon} \edc{\texttt{b} \\ \epsilon} \edc{\texttt{b} \\ \epsilon} \$ \edc{\texttt{a} \\ \epsilon} \edc{\texttt{b} \\ \epsilon} \edc{\texttt{b} \\ \epsilon} \$ \edc{\texttt{a} \\ \texttt{b} \\ \epsilon}^2 \$.
\end{align*}
The second occurrence of \$, i.e., $\eds{S}[4] = \$$ has the LPF $(\texttt{\$ab})^8$ starting at $S[1]$. This length is $(k+1)(f+\ell)+1$, and thus $\texttt{ab}$ must be a common subsequence of $S_1$ and $S_2$ of length $k = 2$.
\end{example}

\section{Hardness of Inequality?}\label{secInequality}

Given two ED strings $\eds{S}$ and $\eds{T}$, we want to address the problem to decide whether $\Language(\eds{S}) \neq \Language(\eds{T})$.
If both strings are GD strings, then each symbol stores strings of equal length.
It is therefore possible to construct a finite deterministic automaton for each symbol with one starting and accepting state, and connect starting and ending states of subsequent symbols to form
an automaton for a GD string~$\eds{S}$ that accepts $\Language(\eds{S})$, cf.~\cite[Sect.~3]{alzamel18degenerate}.
For instance in \cref{figEDAutomaton}, the state $q_4$ is the accepting state of the first symbol and the starting state of the second symbol.
Finally, equivalence of two finite deterministic automata can be checked in time linear in their sizes~\cite{hopcroft71linear}.
The same construction for a (regular) ED string\footnote{In this section, we here use sometimes the adjective \emph{regular} to emphasize on the fact that we consider an ED-string, not a specific subclass like an indeterminate string nor a superclass.} may create a finite nondeterministic automaton because the length difference in an ED symbol can cause $\epsilon$-transitions,
cf.~\cref{figEDAutomaton}.
Testing inequivalence of two nondeterministic finite automata is PSPACE-complete~\cite[AL1]{garey79computers}
while testing inequivalence of regular expressions without Kleene star is NP-complete~\cite[Problem 3.3.10]{garey79computers}\cite[Thm.~2.7]{hunt76equivalence}. 
Since the latter type of expressions is a superclass of ED strings, we know that inequality testing for ED strings is in NP.

\begin{figure}[h]
    \begin{minipage}{0.6\linewidth}
    \adjustbox{valign=c}{
        \(
            \eds{S_1} = \edc{\texttt{aba} \\ \texttt{abb} \\ \texttt{aaa} } \edc{\texttt{a} \\ \texttt{b} }
        \)
    }
    \hspace{1em}
    \includegraphics[valign=c,width=0.5\textwidth,page=1]{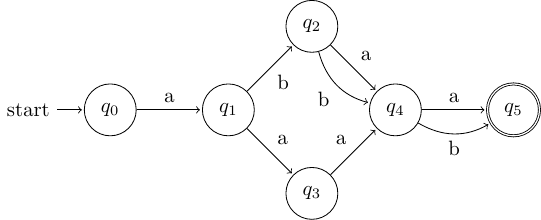}

    \adjustbox{valign=c}{
        \(
            \eds{S_2} = \edc{\texttt{ab} \\ \texttt{a} } \edc{\texttt{a} \\ \epsilon }
        \)
    }
    \hspace{1em}
    \includegraphics[width=0.5\textwidth,page=2]{./img/automata.pdf}
    \end{minipage}
    \begin{minipage}{0.38\linewidth}
    \caption{Automata representing the GD string $\eds{S_1}$ (top) and the ED string $\eds{S_2}$ (bottom). 
    While it is possible to construct an automaton in the size of an ED string, the automaton for ED strings has $\epsilon$-transitions which may blow up the size exponentially when transforming the nondeterministic automaton to a deterministic one.}
    \label{figEDAutomaton}
    \end{minipage}
\end{figure}

Based on the proof of \cite[Thm.~2.7]{hunt76equivalence}, 
we further can show that inequivalence testing is NP-complete for a slight extension of ED strings allowing one level of nesting, which we call 2ED strings.
A \emph{2ED string} is a string $\eds{T}$ such that each symbol $\eds{T}[i]$ is a set of ED strings.
To this end, we reuse the reduction from 3-SAT described \cref{secMUSMAW},
in particular the family of indeterminate strings $\{\eds{T_i}\}_{i=1}^m$ constructed in the proof of \cref{thm_mus_np_hard}.
We define the two 2ED strings $\eds{S}$ and $\eds{T}$ to compare with as

\begin{minipage}{0.4\linewidth}
\[
    \eds{S} = \edc{\eds{T_1}\\ \vdots\\ \eds{T_m}}
    \text{~and~}
    \eds{T} = \edc{0 \\ 1}^n.
\]
\end{minipage}
\begin{minipage}{0.55\linewidth}
We observe that $\eds{T}$ is a (regular) ED string with $\Language(\eds{T}) = \{0,1\}^n$, i.e., the language of $\eds{T}$ generates all binary strings of length $n$.
Finally, $\Language(\eds{S}) \neq \Language(\eds{T})$ if and only if there is a satisfying assignment.
\end{minipage}
However, the complexity of deciding whether $\Language(\eds{S}) \neq \Language(\eds{T})$ for two (regular) ED strings is left open for future work.

\section{Conclusion}

We have studied various problems concerning the characterization of ED strings.
While notably longest common extension queries and the computation of the longest repeating factor can be done in quadratic time,
other problems are NP-hard to compute for specific generalizations of strings such as indeterminate strings.
Since every indeterminate string is a GD string, and every GD string is also an ED string, all hardness results are also true for taking ED strings as an input.

Finally, the studied NP-hard problems are not only NP-hard, but also NP-complete.
To see that, we show that we can verify a certificate in polynomial time.
In all problems, a certificate is a classic string $X$ that solves the addressed problem.
To verify that $X$ is a solution, we apply a pattern matching algorithm on the input ED string~\eds{S}, which runs in polynomial time.
For instance, a certificate for a unique substring or an absent word is a string $X$, for which we check that it occurs exactly once or zero times, respectively.

For future work, additionally to the inequality problem of \cref{secInequality}, we would like to study whether some problems can further parameterized to obtain FPT algorithms, which hopefully work well in practice.

\bibliographystyle{plain}

\end{document}